\documentclass[conference]{IEEEtran}

\usepackage[applemac]{inputenc}
\usepackage[english,american]{babel}
\usepackage{amsmath}
\usepackage{amssymb}
\usepackage{amsthm}
\usepackage{thmtools}
\usepackage{siunitx}
\usepackage{cite} 
\usepackage{tikz}
\usepackage{pgfplots}
\usepackage{graphicx}
\ifCLASSOPTIONcompsoc
\usepackage[caption=false,font=normalsize,labelfont=sf,textfont=sf]{subfig}
\else
\usepackage[caption=false,font=footnotesize]{subfig}
\fi
\usepackage{ctable}
\usepackage{comment}

\newcommand{\Sv}{\mathbf{S}}
\newcommand{\onev}{\mathbf{1}}

\newcommand{\tbs}{t_\mathrm{bs}}
\newcommand{\tn}{t_\mathrm{d}}
\newcommand{\TDW}{\overline{T}_\mathrm{dw}}
\newcommand{\TDD}{\overline{T}_{\eta }}

\IEEEoverridecommandlockouts

\title{MDS-Coded Distributed Storage\\ for Low Delay Wireless Content Delivery}
\begin{document}
\author{
\IEEEauthorblockN{Amina Piemontese and Alexandre Graell i Amat}
\IEEEauthorblockA{Department of Signals and Systems, Chalmers University of Technology, Gothenburg, Sweden}
\thanks{Amina Piemontese is supported by a Marie Curie fellowship (contract 658785-DISC-H2020-MSCA-IF-2014). This work was also was partially funded by the Swedish Research Council under grant \#2011-5961.}}
\maketitle

\begin{abstract}
We address the use of maximum distance separable (MDS) codes for distributed storage (DS) to enable efficient content delivery in wireless networks. Content is stored in a number of the mobile devices and can be retrieved from them using device-to-device communication or, alternatively, from the base station (BS). We derive an analytical expression for the download delay in the hypothesis that the reliability state of the network is periodically restored. Our analysis shows that MDS-coded DS can dramatically reduce the download time with respect to the reference scenario where content is always downloaded from the BS. 
\end{abstract}

\section{Introduction}

The proliferation of mobile devices and the surge of a myriad of multimedia applications has resulted in an exponential growth of the mobile data traffic. In this context, 
wireless caching has emerged as a powerful technique to overcome the backhaul bottleneck, by reducing the backhaul rate and the delay in retrieving content from the network. The key idea is to store popular content closer to the end users. In \cite{Shan13}, a novel system architecture named \emph{femtocaching} was proposed. It consists of deploying a number of small base stations (BSs) with large storage capacity, in which content is stored during periods of offpeak traffic. The mobile users can then download the content from the small BSs, which results in a higher throughput per user. In \cite{Gol14}, it was proposed to store content directly in the mobile devices.  Users can then retrieve content from neighboring devices using device-to-device (D2D) communication or, alternatively, from the serving BS. 

In both scenarios content may be stored using an erasure correcting code, which brings gains with respect to uncoded caching. The use of erasure correcting codes establishes an interesting link between distributed caching for content delivery and distributed storage (DS) for data storage. The key difference is that in the wireless network scenario, data can be downloaded from the storage nodes (the small BSs or the mobile devices) but also from a serving BS, which has always the content available. Therefore, the reliability requirements in DS for data storage can be relaxed. In \cite{Bio15}, the placement of content encoded using a maximum distance separable (MDS) code to small BSs was investigated and it was shown that the backhaul rate can be significantly reduced. In \cite{Paa13}, for the scenario where content is stored in the mobile devices, the repairing of the lost data when a device storing data leaves the network was considered. Assuming  instantaneous repair, the communication cost of data download and repair was investigated. In \cite{Ped15,Ped16}, a repair scheduling where repair is performed periodically was introduced and analytical expressions for the overall communication cost of content download and data repair as a function of the repair interval were derived. Using these expressions, the communication cost entailed by DS using MDS codes, regenerating codes \cite{Dim10}, and locally repairable codes \cite{Pap14} was evaluated in \cite{Ped15,Ped16}.

In this paper, as in \cite{Paa13,Ped15,Ped16}, we consider the scenario where content is stored in the mobile devices, which arrive and depart from a cell according to a Poisson random process. In particular, we assume that content is stored using MDS codes. Our focus is on the delay of retrieving content from the network, which was not considered in \cite{Paa13,Ped15,Ped16}. We derive analytical expressions for the download delay and show that MDS-encoded DS can significantly reduce the delay with respect to the case where content is solely downloaded from the BS.


\section{System Model}\label{s:system_model}
We consider a single cell in a cellular network where mobile devices, referred to as nodes, roam in and out according to a Poisson random process and request a single file at random times. The file is stored in a number of the mobile devices using an MDS code. A copy of the file is also available at the BS serving the cell. 

Nodes arrive according to a Poisson process with exponential independent, identically distributed (i.i.d.) random inter-arrival time $T_a$ with probability density function (pdf)
\begin{equation}\nonumber
f_{T_a}(t)=M\lambda e^{-M\lambda t}, \qquad \lambda\geq 0, t\geq 0,
\end{equation}
where $M\lambda$ is the expected arrival rate of a node and $t$ is time, measured in time units (t.u.).  The nodes stay in the cell for an i.i.d. exponential random lifetime $T_ \ell$ with pdf
\begin{equation}\nonumber
f_{T_\ell}(t)=\mu e^{-\mu t}, \qquad \mu\geq 0, t\geq 0,
\end{equation}
where $\mu$ is the expected departure rate of a node. We assume that $\mu$=$\lambda$, which implies that the expected number of nodes in the network is $M$.

We assume that nodes request the file at random times with i.i.d. random inter-request time $T_r$ with pdf
\begin{equation}\label{e:request time}
f_{T_r}(t)=\omega e^{-\omega t}, \qquad \omega\geq 0, t\geq 0,
\end{equation}
where $\omega$ is the expected request rate per node. 

The file is partitioned into $k$ packets, called symbols, and is encoded into $n$ coded symbols using an $(n,k)$ MDS erasure correcting code of rate $R$ = $k/n$. The encoded data is stored into $n$ nodes, referred to as storage nodes, and hence each storage node stores one symbol. In the rest of the paper, for ease of language, we will sometimes refer to the set of storage nodes as the DS network. For simplicity, we assume $n\ll M$, hence the probability that the number of nodes in the cell is smaller than $n$ is negligible.

In this work we focus on the download process. Each node in the cell can request the file and attempts to retrieve it from the DS network using D2D communication. If the file cannot be completely retrieved from the storage nodes, the BS assists in providing the missing coded symbols. Thanks to the MDS property, the file can be recovered collecting any subset of $k$ coded symbols. The download of a symbol from a storage node incurs $\tn$ t.u., and from the BS $\tbs$ t.u.. We assume that $\tbs\gg \tn$ due to the congestion of the BS-to-node link and the fact that D2D communication occurs over a better channel due to the reduced distance between the involved nodes. We further assume that only one D2D link at a time can be established, and that the D2D communication does not interfere with the communication between the BS and the nodes. Therefore, if the DS network is not idle, the whole file is downloaded from the BS. Moreover, to simplify the analysis, we assume that multiple BS-to-node links can coexist.
Here we do not address the repair problem of restoring the initial state of reliability of the DS network when storage nodes leave the cell~\cite{Ped15,Ped16}. In particular, we assume that the BS keeps track of the storage nodes and repair is performed every $\Delta$ t.u. and, for simplicity, we assume that incurs no transmission delay. Alternatively, we can also assume that nodes arriving in the cell can bring content. This corresponds to the case where the same content is also stored in mobile devices in adjacent cells. Nevertheless, they do not join the DS network instantaneously, but the BS serving the cell periodically updates and broadcasts the list of storage nodes every $\Delta$ t.u.. Therefore, our model complies with both cases.  Similar to~\cite{Ped15,Ped16}, the parameter $\Delta$ is referred to as the repair interval in the sequel.


\section{Average File Download Delay}\label{s:download}
Our performance measure is the download time, referred to as download delay. A node which requests the file is allowed to use D2D communication only if the DS network is idle. Therefore, we introduce the binary random variable (RV) $I\in\{0,1\}$ which describes the status of the DS network. $I=1$ if the network is idle and $I=0$ otherwise. 
If the DS network is idle, the requesting node uses D2D communication to download as many coded symbols as possible (up to $k$) from the DS network and turns to the BS to recover possible missing symbols. If the DS network is occupied, the file is entirely downloaded from the BS.  The average file download delay, $\TDW$, can then be computed as
\begin{equation}\label{e:Tdw}
\TDW=   \text{Pr}\{I=1\} (  \TDD + (k-\eta) \tbs  )  +    \text{Pr}\{I=0\} k \, \tbs \, ,
\end{equation}
where $\eta$ is the average number of coded symbols downloaded using D2D communication per request and $\TDD$ is the corresponding delay. In the following, $\TDD$ is referred to as the average D2D download delay. 

The first step in our derivation is the computation of the probability that the DS network is idle. Let $I^{(\ell)}$ be the status of the network at the time of the $\ell$th request. We have
\begin{equation}\label{e:idle}
\Pr\{ I=1\} =\lim_{L \to \infty} \frac{1}{L} \sum_{\ell=1}^L \Pr\{ I^{(\ell)}=1\}.
\end{equation}
In order to compute $\Pr\{ I^{(\ell)}=1\}$, we introduce the RV $W^{(j)}$ that denotes the time instant of the $j$th request. Also, let $T^{(j)}$ be the time during which the DS network is occupied by the $j$th request. The DS network is idle at the time of the  $\ell$th request if none of the previous requests is still using D2D communication. Therefore, $\Pr\{ I^{(1)}=1\}=1$ and 
\begin{equation}\label{e:P_I exact}
\Pr\{ I^{(\ell)}=1\}\!=\!\prod_{i <\ell} \! \Pr\{  W^{(\ell)} \!\!>  \!\!W^{(\ell-i)} \!+ T^{(\ell-i)}  \}, \,\ell>1 .
\end{equation}
Assuming that if the DS network is occupied at time $W^{(\ell)} $ is because of the $(\ell-1)$th request, the product in (\ref{e:P_I exact}) reduces to the term involving the $(\ell-1)$th request only, i.e.,
\begin{align}\label{e:P_I approx1}
\Pr\{ I^{(\ell)}=1\}&\simeq \Pr\{  W^{(\ell)} >  W^{(\ell-1)} + T^{(\ell-1)}  \} \\
&= \int_0^\infty \Pr\{  W^{(\ell)} >  W^{(\ell-1)} + t  \} f_{T^{(\ell-1)} }(t) d t\, ,\nonumber
\end{align} 
where $ f_{T^{(\ell-1)} }(t)$ is the pdf of  $T^{(\ell-1)} $. Since the requests are i.i.d. with inter-request time distributed as in (\ref{e:request time}) and on average there are $M$ nodes in the cell,  we can compute
\begin{equation}\nonumber
\Pr\{  W^{(\ell)} >  W^{(\ell-1)}+ t  \} = e^{-\omega M t}\, ,\quad t\geqslant 0,\quad\ell>1\, ,
\end{equation}
and (\ref{e:P_I approx1}) can be written as 
\begin{equation}\nonumber
\Pr\{ I^{(\ell)}=1\}\simeq \mathbb{E}_{T^{(\ell-1)}} \{  e^{-\omega M T^{(\ell-1)}}  \} ,\quad\ell>1,
\end{equation}
where $\mathbb{E}_x\{\cdot\}$ represents the expectation with respect to the variable $x$. If $\omega T^{(\ell-1)} \ll 1$,
\begin{equation}\label{e:P_I approx2}
e^{-\omega M T^{(\ell-1)}} \simeq 1-\omega M T^{(\ell-1)}
\end{equation}
and
\begin{align}
\Pr\{ I^{(\ell)}=1\} & \simeq  \mathbb{E}_{T^{(\ell-1)}} \{  e^{-\omega M T^{(\ell-1)}}  \} \simeq   \nonumber\\
                           & \simeq  \mathbb{E}_{T^{(\ell-1)}} \{   1-  \omega M T^{(\ell-1)}   \} \nonumber\\
                           & = 1-\omega M  \TDD \Pr\{ I^{(\ell-1)}=1\} .\label{e:idle_ell}
\end{align}
In~(\ref{e:idle_ell}), we have used the fact that the average D2D download delay is independent of the specific request (if $\ell$ is sufficiently large). This result is proven in Lemma~\ref{lemma}. Substituting (\ref{e:idle_ell}) in (\ref{e:idle}) and after some simple calculations, we obtain
\begin{equation}\nonumber
\Pr\{ I=1\} = \frac{1}{1+\omega M \TDD} \, .
\end{equation} 
Note that in the expression above, with some abuse of notation, we use equal sign to avoid carrying all the way the approximation sign due to the approximations introduced in (\ref{e:P_I approx1}) and (\ref{e:P_I approx2}). 
 
We now consider the computation of the average D2D download delay and the average number of coded symbols downloaded using D2D per request. We assume that a node cannot download in parallel from multiple nodes, but it serially tries to download $k$ symbols from the DS network. When a node requests the file, if the DS network is idle, it randomly choses one of the storage nodes from the list supplied by the BS. After each downloaded symbol, the requesting node randomly choses the next storage node among those belonging to the list and still alive.
We assume that a requesting node that has collected less than $k$ symbols turns to the BS when all the reference storage nodes  left or when the download of a symbol fails, even if other storage nodes are available. To simplify the analysis, we assume that both cases (the failed symbol download and the absence of storage nodes) incur $\tn$ t.u., even if the node could contact the BS earlier. We also assume that the download from the DS network fails if the requesting node itself leaves the cell before collecting $k$ symbols. In this case, the download is also completed from the BS. 

The download from the storage nodes can be fully successful or only partially accomplished. In order to describe the D2D download, we define $S_1$ the binary RV which describes the success of download at the first attempt. More precisely, $S_1=1$ represents the successful download of the coded symbol from the first contacted storage node. If download is not successful from the first contacted storage node, $S_1=0$. Similarly, we define $S_j$ the binary RV describing the download at the $j$th attempt and we denote by $\Sv_{[i]}$, $i\ge 1$ the random vector  ($S_1,..., S_i$). In the following, $\onev_{j}$ represents the all-ones vector of length $j$.

According to our model, the requesting node completes the download of $k$ symbols from the DS network in $k \tn$ t.u. with probability $\Pr\{ \Sv_{[k]}=\onev_k \} $, while the partial download of $j<k$ symbols happens with probability \mbox{$\Pr \{\Sv_{[j]}=\onev_{j}, S_{j+1}=0 \}$} and incurs $(j+1)\tn$ t.u.. In the computation of the average D2D download delay, we also consider the case where download from the DS network completely fails. The corresponding probability is $ \Pr \{S_1=0\}$ and the delay is $\tn$. 
To simplify the analysis, we do not take into account the fact that the request may originate from a storage node, i.e., we do not consider that a storage node needs to download $k-1$ symbols instead of $k$.
Finally, the average D2D download delay $\TDD$ and the average number of D2D downloaded symbols $\eta$ are given by
\begin{align}
\eta= &k \Pr\{ \Sv_{[k]}=\onev_k \}  + \sum_{j=1}^{k-1} j   \Pr \{\Sv_{[j]}=\onev_{j}, S_{j+1}=0 \} \nonumber\\
\TDD= & \tn \Big( \eta +  \Pr \{S_1=0\} \!+  \!\! \sum_{j=1}^{k-1}   \Pr \{\Sv_{[j]}=\onev_{j}, S_{j+1}=0 \}\Big).  \nonumber  
\end{align}
In the next section, we derive $\Pr \{S_1=0\}$, $\Pr\{ \Sv_{[k]}=\onev_k \}$, and $\Pr \{\Sv_{[j]}=\onev_{j}, S_{j+1}=0 \}$.


\section{Probability of D2D Download}\label{s:D2D}
In this section, we derive the probability that the content is fully recovered from the DS network, $\Pr\{ \Sv_{[k]}=\onev_k \} $, the probability that it is only partially recovered, $\Pr \{\Sv_{[j]}=\onev_{j}, S_{j+1}=0 \}$, and the probability that no symbols can be downloaded from the DS network, $ \Pr \{S_1=0\}$. We also show that the average D2D download time $\TDD^{(\ell)} \triangleq\mathbb{E} \{ T^{(\ell)} \}$  does not depend on the specific request if $\ell$ is sufficiently large. 

We introduce of the following RVs and events.
\begin{itemize}
\item $S_i^{(\ell)}\in\{0,1\}$ is the binary RV describing the successful symbol download at the $i$th attempt of the $\ell$th request. 
\item $X_i^{(\ell)}\in\{0,\ldots,n\}$ is the number of storage nodes available at the time of the $i$th attempt of the $\ell$th request, i.e., the available storage nodes not yet contacted. In~\cite{Ped15}, it was shown that the probability that there are $x_1$ storage nodes at the instant of the $\ell$th request, $X_1^{(\ell)}=x_1$, does not depend on $\ell$ (when $\ell$ grows large), and is given by
\begin{align}
&\Pr\{ X_1^{(\ell)}=x_1\}=\nonumber \\
&\frac{1}{\Delta} \sum_{i'=x_1}^n \!\!\!\frac{1-p_{i'}}{\mu_{i'}} \!\!\prod_{{\substack{ j=x_1 \\ j\neq i' }} }^n \frac{j}{j-i'}- \frac{1}{\Delta} \!\sum_{i'=x_1\!+\!1}^n \!\!\frac{1-p_{i'}}{\mu_{i'}} \!\!\prod_{{\substack{ j=x_1\!+\!1 \\ j\neq i' }} }^n \frac{j}{j-i'}\, ,\nonumber
\end{align}
where $\mu_{i'}=i'\mu$ and $p_i'=e^{-\mu_{i'}\Delta}$. 

To ease notation in the remainder of the paper, we define $h(x_1)\triangleq \Pr\{ X_1^{(\ell)}=x_1\}$.
\item $F_i^{(\ell)}  \in \{0,\ldots,n\}$ is the number of departures in $\tn$ t.u. among the $X_i^{(\ell)}$ storage nodes available at the time of the $i$th attempt of the $\ell$th request. We are interested in the probability $\Pr\{ F_i^{(\ell)}=f | X_i^{(\ell)}=x\}$. Its derivation is similar to that of $\Pr\{X_1^{(\ell)}=x_1 \}$. We obtain
\begin{align}\nonumber
&\Pr\{ F_i^{(\ell)}=f | X_i^{(\ell)}=x\}=\nonumber\\
&\sum_{i'=x-f}^x \!\!\! e^{-\mu_{i'} \tn} \!\!\! \prod_{{\substack{ j=x-f \\ j\neq i' }} }^x \frac{j}{j-i'}- \! \!\! \! \! \sum_{i'=x-f+1}^x \!\!\!  e^{-\mu_{i'} \tn} \!\!\! \!  \prod_{{\substack{ j=x-f+1 \\ j\neq i' }} }^x \frac{j}{j-i'}\, .\nonumber
\end{align}
The probability above is independent of $\ell$ and $i$ and we define $g(f,x)\triangleq \Pr\{ F_i^{(\ell)}=f | X_i^{(\ell)}=x\}$. It follows that $g(f,x)=0$ if $f>x$ and $g(0,0)=1$.
\item $D^{(\ell)}$ is the departure time of the node which places the $\ell$th request.
\item $\mathcal{A}_i^{(\ell)}=\{  D^{(\ell)} -  W^{(\ell)} > i \tn\}$ is the event that the node which places the $\ell$th request stays in the network for more than $ i \tn $ t.u. from the start of the download. The corresponding probability does not depend on $\ell$ and is given by
\begin{equation}\nonumber
\Pr\{ \mathcal{A}_i^{(\ell)}\} =e^{-i\mu \tn}\, .
\end{equation}
We define $a_i\triangleq \Pr\{ \mathcal{A}_i^{(\ell)}\}$.
\item $\mathcal{B}_i^{(\ell)}=\{ (i-1)\tn < D^{(\ell)} -  W^{(\ell)}  < i \tn\}  $ is the event that the node which places the $\ell$th request departs after the $(i-1)$th download attempt but before the $i$th one. The probability of this event is
\begin{equation}\nonumber
\Pr\{ \mathcal{B}_i^{(\ell)}\} =e^{-(i-1)\mu \tn}(1-e^{-\mu \tn})\, 
\end{equation}
and is independent of $\ell$. We define $b_i\triangleq \Pr\{ \mathcal{B}_i^{(\ell)}\}$.
\end{itemize}
In the following, the probabilities $\Pr\{ S_j^{(\ell)} |   X_j^{(\ell)},  F_j^{(\ell)},  \mathcal{A}_j^{(\ell)}\}$ are computed by ignoring the fact that the download request may originate from a storage node itself. The goodness of this approximation for the considered scenarios has been validated through computer simulations.
\subsection{No Symbol is Downloaded}
We first consider $\Pr\{ S_1=0\}$, which is given by
\begin{equation}\nonumber
\Pr\{ S_1=0\} =\lim_{L \to \infty} \frac{1}{L} \sum_{\ell=1}^L \Pr\{ S_1^{(\ell)}=0\}\, ,
\end{equation}
and compute $\Pr\{ S_1^{(\ell)}=0\}$. The recovery of the first symbol fails if the requesting node leaves the cell before completing the download. It also fails if the requesting node stays in the cell but no storage nodes are available or if it chooses to download from a storage node which departs before $t_n$ t.u. from the start of the download. Therefore,
\begin{align}
&\Pr\{ S_1^{(\ell)}=0\}= \nonumber \\
&\Pr\{ S_1^{(\ell)}=0 | \mathcal{B}_1^{(\ell)}\} \Pr\{ \mathcal{B}_1^{(\ell)}\}  + \Pr\{ S_1^{(\ell)}=0  | \mathcal{A}_1^{(\ell)}\} \Pr\{ \mathcal{A}_1^{(\ell)}\} \nonumber\\
&=b_1  + a_1\sum_{x_1 f_1}  \Pr\{ S_1^{(\ell)}=0  |   X_1^{(\ell)}=x_1,  F_1^{(\ell)}=f_1,  \mathcal{A}_1^{(\ell)}\} \cdot \nonumber \\
 &\qquad \qquad \qquad  \cdot \Pr\{F_1^{(\ell)}=f_1,X_1^{(\ell)}=x_1|  \mathcal{A}_1^{(\ell)}\}   \, . \label{e:PS1b}
\end{align}
The joint probability mass function of the number of storage nodes available for download and the number of storage nodes that depart before $W^{(\ell)}+\tn$ is independent of the departure time of the requesting node (note that this is an approximation if the requesting node is a storage node).
Hence, we have
\begin{align}
&\Pr\{F_1^{(\ell)}=f_1,X_1^{(\ell)}=x_1|  \mathcal{A}_1^{(\ell)}\} = \nonumber \\
&\Pr\{F_1^{(\ell)}=f_1 |X_1^{(\ell)}=x_1 \} \Pr\{X_1^{(\ell)}=x_1 \} = h(x_1) g(f_1,x_1)\,.\nonumber
\end{align}
The probability   $\Pr\{ S_1^{(\ell)}=0  |   X_1^{(\ell)}=x_1,  F_1^{(\ell)}=f_1,  \mathcal{A}_1^{(\ell)}\}$ is equal to 1 if there are no storage nodes available, i.e., $x_1=0$. Otherwise, 
it equals the probability to choose one of the $f_1$ storage nodes that leave the cell in $\tn$ t.u., i.e., $f_1/x_1$, with $f_1\leq x_1$. 
Since the probabilities involved in~(\ref{e:PS1b}) are all independent of $\ell$, we finally have
\begin{equation}\nonumber
\Pr\{ S_1=0\} =  b_1  + a_1 h(0) + a_1 \sum_{x_1 =1}^n \sum_{f_1 =0 }^{x_1} \frac{f_1}{x_1}    h(x_1)  g(f_1,x_1)\, .
\end{equation}

\subsection{Partial and Complete Download}
To evaluate the probability that $k$ symbols are downloaded from the DS network, we start with the following limit
\begin{equation}\nonumber
\Pr\{  \Sv_{[k]} =\onev_k \}= \lim_{L \to \infty}   \frac{1}{L} \sum_{\ell=1}^L \Pr\{  \Sv_{[k]}^{(\ell)}=\onev_k \}\, .
\end{equation}
We consider the $\ell$th request and, similarly to the previous case, we will find that this probability is independent of $\ell$. We have
\begin{align}\nonumber
&\Pr\{  \Sv_{[k]}^{(\ell)} =\onev_k \}=\sum_{x f}  \Pr\{  \Sv_{[k]}^{(\ell)}=\onev_k, X_k^{(\ell)}=x,  F_k^{(\ell)}=f \} \\
&=\sum_{x f}  \Pr\{ S_k^{(\ell)}=1  |   X_k^{(\ell)}=x,  F_k^{(\ell)}=f ,   \mathcal{A}_k^{(\ell)} \} a_k \cdot \nonumber \\
& \cdot \Pr\{F_k^{(\ell)}=f,X_k^{(\ell)}=x  , \Sv^{(\ell)}_{1,k-1}=\onev_{k-1}\}  \, .\nonumber
\end{align}
The probability $\Pr\{ S_k^{(\ell)}=1  |   X_k^{(\ell)}=x,  F_k^{(\ell)}=f,  \mathcal{A}_k^{(\ell)}\}$ for $x>0$ and $f<x$ equals the probability to choose one of the storage nodes that stay in the cell, i.e., $\frac{x-f}{x}$.

We now evaluate the probability $\Pr\{F_j^{(\ell)}=f,X_j^{(\ell)}=x  , \Sv^{(\ell)}_{1,j-1}=\onev_{j-1}\}$ for $j>1$, which will be also used for the computation of the probability of partial D2D download. For $f\leq x$, we have the following recursion
\begin{align}
&\Pr\{F_j^{(\ell)}=f,X_j^{(\ell)}=x  , \Sv^{(\ell)}_{1,j-1}=\onev_{j-1}\}= \nonumber\\
&g(f,x) \Pr\{X_j^{(\ell)}=x  , \Sv^{(\ell)}_{1,j-1}=\onev_{j-1}\}=\nonumber\\
&g(f,x)\sum_{x' f'}  \Pr\{X_j^{(\ell)}=x | X_{j-1}^{(\ell)}=x',F_{j-1}^{(\ell)}=f', S^{(\ell)}_{j-1}=1\} \cdot\nonumber\\
&\cdot \Pr\{S_{j-1}^{(\ell)}=1| F_{j-1}^{(\ell)}=f' ,X_{j-1}^{(\ell)}=x'  , \Sv^{(\ell)}_{1,j-2}=\onev_{j-2}\}  \cdot\nonumber\\
&\cdot  \Pr\{F_{j-1}^{(\ell)}=f',X_{j-1}^{(\ell)}=x'  , \Sv^{(\ell)}_{1,j-2}=\onev_{j-2}\}     \, . \label{e:PP}
\end{align}
We define $ N(x,x',f')\triangleq \Pr\{X_j^{(\ell)}=x | X_{j-1}^{(\ell)}=x',F_{j-1}^{(\ell)}=f', S^{(\ell)}_{j-1}=1\}$, which is equal to one if $x=x'-f'-1$ and $x'>f'$, and zero otherwise. The condition $x=x'-f'-1$ follows from the fact that the number of available storage nodes after a successful symbol download is equal to the number of storage nodes still alive, \mbox{$x'-f'$}, minus the storage node just used. The condition $x'>f'$ comes from the fact that the $(j-1)$th symbol download is assumed to be successful.   

It is easy to prove by induction that the probability (\ref{e:PP}) does not depend on $\ell$ . By defining $\gamma_j(x,f)\triangleq\Pr\{F_j^{(\ell)}=f,X_j^{(\ell)}=x  , \Sv^{(\ell)}_{1,j-1}=\onev_{j-1}\}$, we obtain the following recursion for $j>1$,
\begin{equation}\nonumber
\gamma_j(x,f)\!=\!g(f,x) a_{j-1}\!\!\sum_{x'=1}^n \sum_{f'=0}^{x'}  \! \frac{x'-f'}{x'}  N\!(x,x',f') \gamma_{j-1}\!(x',f')
\end{equation}
with initial condition $\gamma_1(x,f)=h(x)g(f,x)$. The probabilities $\gamma_j(x,f)$, $j\geq 1$, are equal to zero for $f>x$.

Finally, the probability of complete download from the DS network is
\begin{equation}\nonumber
\Pr\{  \Sv_{[k]} =\onev_k \}=a_k \sum_{x =1}^n \sum_{f=0 }^{x} \frac{x-f}{x}  \gamma_k(x,f)\, .
\end{equation}

Following a similar approach, we can compute the probability of partial download,
\begin{align}
\Pr \{\Sv_{[j]}=&\onev_{j}, S_{j+1}=0 \} =  \gamma_{j+1}(0,0) +    \nonumber\\
&+\sum_{x=1}^n \sum_{f=0}^x \Big( 1- \frac{x-f}{x}  a_{j+1}\Big)  \gamma_{j+1}(x,f)  \nonumber\, .
\end{align}


The results above allow also to prove the following lemma.
\newtheorem{Lemma}{Lemma}
\begin{Lemma}\label{lemma}
 The average D2D download time for the $\ell$th request, $\TDD^{(\ell)}$, is independent of the specific request if the index $\ell$ is sufficiently large. 
 \end{Lemma}
\begin{proof}
Similarly to the average D2D download delay, $\TDD^{(\ell)}$ is
\begin{align}
\TDD^{(\ell)}= &k \tn \Pr\{ \Sv^{(\ell)}_{[k]}=\onev_k \} + \tn \Pr \{S^{(\ell)}_1=0\} +\nonumber \\  
&+ \sum_{j=1}^{k-1}  (j+1) \tn \Pr \{\Sv^{(\ell)}_{[j]}=\onev_{j}, S^{(\ell)}_{j+1}=0 \} \nonumber\;\; [t.u.]\, .
\end{align}
The Lemma follows from the fact that the probabilities in the expression above are independent of $\ell$, when $\ell$ grows large.  
\end{proof}

\section{Results}
In this section, we consider the performance of a wireless network with $M=30$ nodes, departure rate $\mu=1$, and request rate $\omega=0.02$. 
We compare the average file download delay of the considered network with MDS-coded DS with the delay of the traditional scenario where the content is solely downloaded from the BS. We consider several $(n,k)$ MDS codes and also an uncoded scenario where one storage node in the cell stores the file. We denote by $T_\text{ref}=k \tbs$ the delay incurred in the traditional scenario, and we fix $T_\text{ref}=1$ t.u..

In Figs.~\ref{f:0.1}--\ref{f:0.001}, we show the gain that can be achieved using MDS-coded DS, by reporting the ratio between $T_\text{ref}$ and $\TDW$ as a function of the repair interval. In Fig.~\ref{f:0.1}, \ref{f:0.01}, and \ref{f:0.001}, $\tn$ is 10, 100, and 1000 times, respectively, smaller than $\tbs$. The results clearly show that MDS-coded DS can greatly improve the performance in terms of content download delay, provided that the update interval, $\Delta$, is sufficiently small.
\begin{figure}
\centering{} \includegraphics[width=0.96\columnwidth]{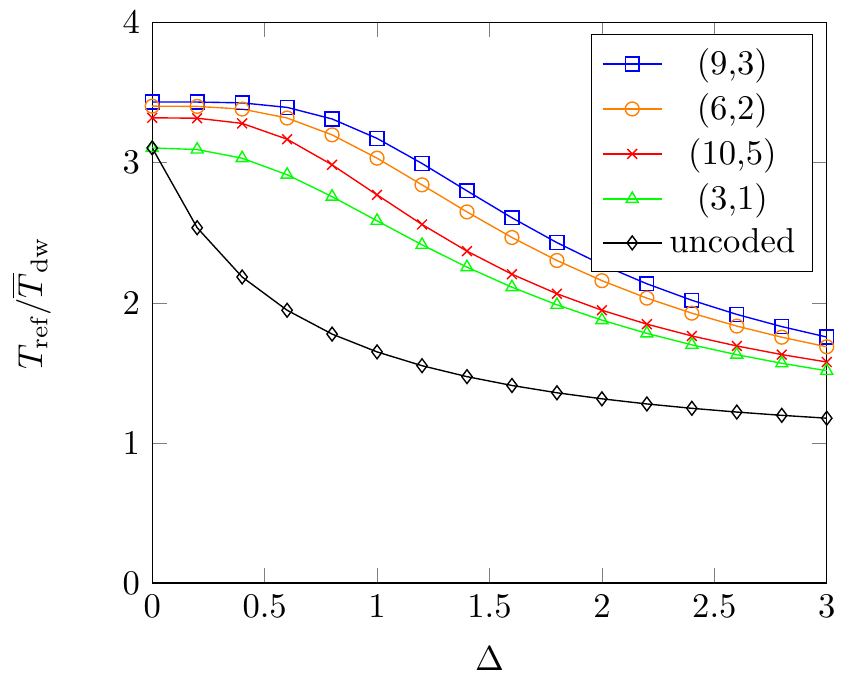}
\vspace{-4mm}
 \caption{Ratio between the file download delay without D2D communication and that of the scenario using MDS-coded DS. $\tbs=10 \tn$.}\label{f:0.1}
 \vspace{-3mm}
\end{figure}
\begin{figure}
\centering{} \includegraphics[width=0.96\columnwidth]{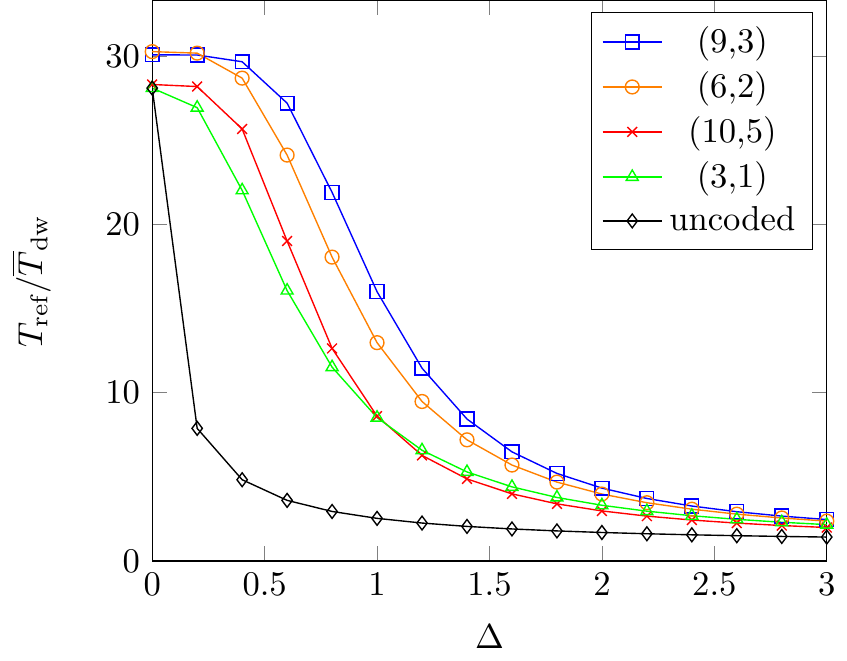}
\vspace{-4mm}
 \caption{Ratio between the file download delay without D2D communication and that of the scenario using MDS-coded DS. $\tbs=100 \tn$.}\label{f:0.01}
  \vspace{-3mm}
\end{figure}
\begin{figure}
\centering{} \includegraphics[width=0.96\columnwidth]{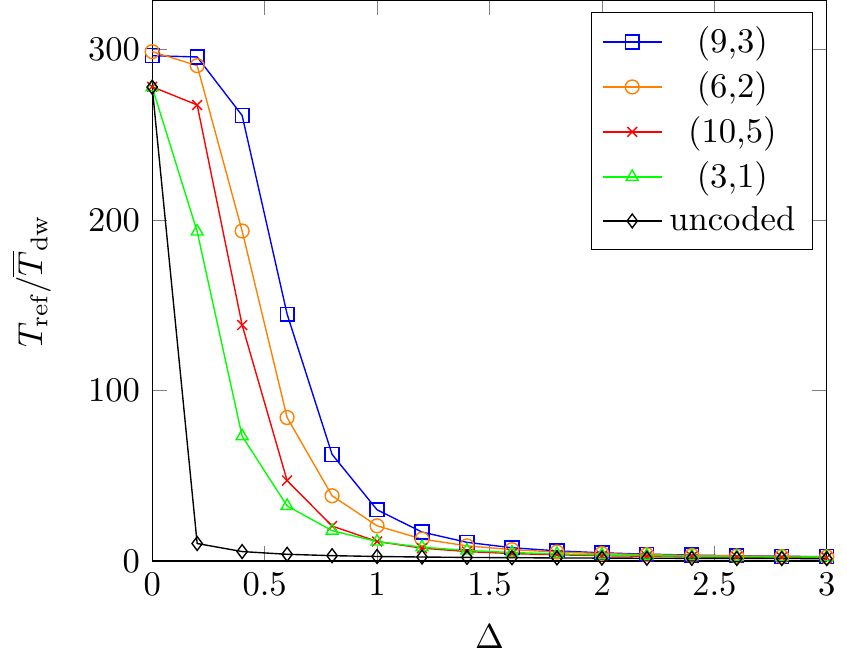}
\vspace{-4mm}
  \caption{Ratio between the file download delay without D2D communication and that of the scenario using MDS-coded DS. $\tbs=1000 \tn$.}\label{f:0.001}
   \vspace{-5mm}
\end{figure}

\section{Conclusions}
We considered the application of MDS-coded distributed storage to wireless networks and computed the average file download delay when users are allowed to use device-to-device communication. MDS-coded DS can dramatically reduce the download delay with respect to the traditional case where content is always downloaded from the base station.

\bibliographystyle{IEEEtran}
\end{document}